\documentclass[12pt]{article}

\usepackage[margin=1in]{geometry}  
\usepackage{graphicx}              
\usepackage{epstopdf}
\usepackage{epsfig}
\usepackage{amsmath}               
\usepackage{amsfonts}              
\usepackage{amsthm}                

\usepackage{xcolor}
\usepackage{color}
\usepackage{colortbl}

\newtheorem{thm}{Theorem}[section]
\newtheorem{lem}[thm]{Lemma}

\newcommand{\sean}[1]{}
\newcommand{\T}{{\bold T}}
\newcommand{\F}{{\bold F}}
\newcommand{\bif}{{\bf if }}
\newcommand{\belse}{{\bf else }}
\newcommand{\bthen}{{\bf then }}
\newcommand{\R}{{\mathcal R}}
\newcommand{\D}{{\mathcal D}}

\providecommand{\norm}[1]{\lVert#1\rVert}

\usepackage[]{algorithm2e}

\begin{document}

\title{Computing \\affine combinations, distances and correlations of \\ recursive partition functions}

\author{Sean Skwerer and Heping Zhang\\ Collaborative Center for Statistics in Sciences\\ Yale University}

\date{November 11, 2014}
\maketitle

\begin{abstract}
Recursive partitioning is the core of several statistical methods including Classification and Regression Trees, Random Forest, and AdaBoost.
Despite the popularity of tree based methods, to date, there did not exist methods for combining multiple trees into a single tree, or methods for systematically quantifying the discrepancy between two trees.
Taking advantage of the recursive structure in trees we formulated fast algorithms for computing affine combinations, distances and correlations in a vector subspace of recursive partition functions.
\end{abstract}

\newpage
\tableofcontents
\newpage

\section{Introduction}
Recursive partitioning is the core for many statistical and machine learning methods including Classification and Regression Trees, Multivariate Adaptive Regression Splines, AdaBoost and Random Forest.
Methods based on recursive partitioning are regarded among the top data mining methods \cite{wu2007}, and in a study of over 100 classification methods versions of Random Forest occupied 3 out of the top five spots \cite{Fernandez-Delgado2014}.

The widespread application of recursive partitioning methods can be attributed to their versatility, speed and robustness.
Recursive partitioning has been used to solve regression, density estimation and classification problems.
Recursive partitioning procedures are divide and conquer algorithms which find optimal partitions of the data at each recursive stage.
The criteria for an optimal partition depends on the type of problem: regression, classification, or density estimation, but in all 
cases the partition is selected to optimize some quantification of purity e.g. reducing variance in regression or homogenizing the distribution in classification or density estimation -- 
for more details see \cite{hastie2009elements,Zhang2010}.
The bifurcating structure created by recursive partitioning algorithms is often referred to as a tree -- we give a formal definition in Section \ref{sec:tree}.

It is common for recursive partitioning algorithms to employ stopping rules aimed at preventing over-fitting such as stopping when the number of observations is below a certain quantity \cite{breiman1984classification} or stopping when the conditional p-value of a partition is to low  \cite{Strobl2009}.
Researchers haven proven that recursive partitioning algorithms will produce trees that approach the true optimal decision rule as more and more data become available \cite{Gordon1978},
however methods which use only a single tree are often out-performed by ensemble methods which use an average or mode of the predictions from a collection of trees.
Currently, the two most widely used methods to build ensembles of trees are random forests and boosting.

A complete theory for ensembles of trees has been a topic of research in statistics and machine learning since their introduction when they achieved state of the art performance in classification problems \cite{freund1996,Breiman1996,Breiman2001}.
An explanation of boosting as optimization in function space \cite{Friedman2000}, 
as opposed to parameter space, has lead to 
(1) consistency results achieved through regularization techniques which were not part of the original algorithm \cite{Zhang2005}, and 
(2) new algorithms which minimize other cost functionals and in some cases improve on the original boosting algorithm \cite{Mason2000,Buhlmann2003}.
Considerable progress in theory for random forests has been made recently 
\cite{Wager2015a,Scornet2015,Mentch2014}. 
These represent a frontier in analyses of random forests where the trend has been assumptions  in the hypotheses of theorems which are more closely aligned with practice and more sophisticated results leading to estimators of prediction variance. 
All things considered this progress is a big step forward, and will contribute to methods of inference based on quantification of variation in point estimates from ensembles.
However, the existing theory is not complete.
The need to understand how ensembles of trees are able to avoid over-fitting by a mechanism which complements stopping rules and pruning has been emphasized  \cite[Discussion: Breiman]{Friedman2000}
 and recently an explanation of how over-fit trees can form a reliable ensemble has been offered
 \cite{Wyner2015}.
Theory for ensembles of trees is continuing to mature, and hopefully some unifying understanding will come soon.

Recursive partitioning algorithms select optimal partitions at each stage, but generally there are no guarantees about their global optimality with respect to the trade off between purity and number of partitions, or expected error for predictions.
Bayesian methods for building trees, which have been shown to outperform recursive partitioning in some cases, apply sophisticated optimization procedures which employ Markov-Chain-Monte-Carlo search \cite{Chipman2010,Chipman2012}.
An affine combination of trees is a tree, and thus ensemble methods actually produce a single tree represented as a weighted sum of many trees.
Therefore ensemble methods can be interpreted as better algorithms for finding an optimal tree, even if they do not explicitly give a single tree.
This begs several questions: Are all the trees in an ensemble necessary? Is it possible to produce one tree or a few trees which can perform approximately as well as an entire ensemble?
What is the trade-off between parsimony and predictive power of an ensemble?
To this end we created an algorithm for computing a tree which represents an affine combination of trees.
We extend this method to compute quantities which measure the similarity or difference between trees.
These measures can be used to explore and summarize the distribution of trees in an ensemble via multi-dimensional scaling or cluster analysis.
To measure how well one forest approximates another, we introduce a method for computing the distance between two forests.

The remaining content of this manuscript is organized as follows.
The preliminary section focuses on the main element of interest: trees obtained from recursive partitioning algorithms.
The basics of recursive partitioning are discussed in Section \ref{sec:introRP}.
The result of recursive partitioning algorithms is a tree, which we define in Section \ref{sec:tree}.
A generic version of our algorithm for combining trees is presented in Section \ref{sec:CombineRPF}.
In Section \ref{sec:distances} we describe distances and correlations for the functions defined by trees.
Implementation details for specific types of trees are discussed in Section \ref{sec:specifications}.

\section{Preliminaries}

\subsection{Introduction to recursive partitioning}\label{sec:introRP}
We introduce recursive partitioning with an example from Chapter 2 of \cite{Zhang2010}, which uses the database from the Yale Pregnancy Outcome Study. 
In this example a subset of 3,861 women whose pregnancies ended in a singleton live birth are selected from this database. 
Preterm delivery is the outcome variable of interest, and 15 variables are candidates to be useful in representing routes to preterm delivery. 
The candidate predictor variables are listed in Table \ref{table:PretermBirthVariables}.

\begin{table}
	\begin{tabular}{l | l l l }
		\hline
		\hline
		Variable name & Label & Type & Range/levels\\
		Maternal age & $x_1$ & Continuous & 13-46\\
		Marital status & $x_2$ & Nominal & Currently married,\\
		& & & divorced, separated,\\
		& & & widowed, never married\\
		Race & $x_3$ & Nominal & White, Black, Hispanic,\\
		& & & Asian, Others \\
		Marijuana use & $x_4$ & Nominal & Yes, no \\
		Times of using marijuana & $x_5$ & Ordinal $ \geq 5, 3-4, 2, 1 (daily)$\\
		Years of education & $x_ 6$ & 4-27\\
		Employment & $x_7$ & Nominal & Yes, no\\
		Smoker & $x_8$ & Nominal & Yes, no\\
		Cigarettes smoked & $x_9$ & Continuous & 0-66\\
		Passive smoking & $x_{10}$ & Nominal & Yes, no\\
		Gravidity & $x_{11}$ & Ordinal & 1-10\\
		Hormones/DES used by mother & $x_{12}$ & Nominal & None, hormones, DES,\\
		& & & both, uncertain\\
		Alcohol (oz/day) & $x_{13}$ & Ordinal & 0-3\\
		Caffeine (mg) & $x_{14}$ & Continuous & 12.6-1273\\
		Parity & $x_{15}$ & Ordinal & 0-7
		
	\end{tabular}
	\caption{Candidate predictor variable for preterm delivery from the Yale Pregnancy Outcome Study database.}
	\label{table:PretermBirthVariables}
	
\end{table}

\begin{figure}
	\includegraphics[width=9cm]{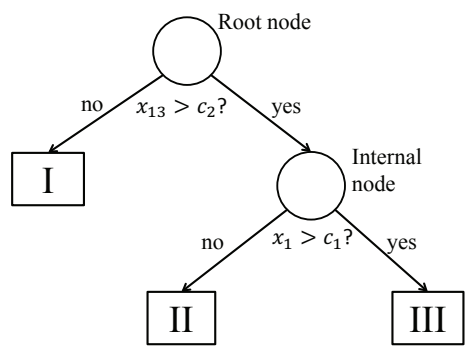}
	\caption{An illustrative tree structure for pathways to preterm birth. $x_1$ is age and $x_{13}$ is the amount of alcohol drinking. Circles and dots represent different outcomes. }
	\label{fig:PretermBirthTree}
\end{figure}

Consider the example tree diagram in Fig. \ref{fig:PretermBirthTree}. 
The tree has three layers of nodes. The first layer contains the unique root node, which is the circle at the top of the tree. 
The root node is partitioned into two daughter nodes in the second layer of the tree: one terminal node, which is the box marked I, and one internal node, namely the circle down and to the right of the root node. 
The internal node is partition into two daughters, which are the terminal nodes marked II and III. 

The tree represents a recursive partition of the data.
Recursive partitioning begins after the root node, which contains all the data. 
Moving down from the root node to the second layer data are partitioned to the right daughter if $x_{13}>c_2$ and partitioned to the left daughter if $x_{13}\leq c_2$. 
Thus all data with $x_{13}\leq c_2$ are contained in the terminal node labeled I. 
On the other hand, moving down from the internal node to the third layer, data with $x_{1} > c_1$ are partitioned to the right daughter, and data with $x_{1}\leq c_1$ are partition to the left daughter.
Thus data at terminal nodes marked II and III are recursively partitioned first at the root node and then again at the internal node.

Generally classification and regression trees can have many layers. 
Algorithms are used to construct trees.
During tree construction two main decisions are made why and how a parent node is split into two daughter nodes and when to declare a terminal node. 
Criterion to make these decisions are based on homogeneity of the data in a node. 
There are several methods, for a full treatment see \cite{Zhang2010}.
In the next section we formally define a representation of recursive partition functions.

\subsection{Trees}\label{sec:tree}
A \emph{partition} of a set $A$ is a collection of subsets of $A$ which form a disjoint cover of $A$.
A binary tree $T$ is a list of nodes $v_0,\ldots,v_m$ which have following attributes.
\begin{itemize}
	\item[] {\bf Internal or terminal}: Every node $v$ has a parent node $p(v)$, except one node which has no parent, called the root. 
	A node, $v$, is internal if it has a left daughter $l(v)$ and a right daughter $r(v)$, which are elements of the set $\{v_1,\ldots,v_m\}$, otherwise
	$v$ has no daughters, and is called a terminal node.
	\item[] {\bf Parents} For any two nodes $u$ and $v$ if $u$ is the parent of $v$, $p(v)=u$, then $v$ must be a daughter of $u$, that is either (i) $l(u)=v$ or (ii) $r(u)=v$.
	\item[] {\bf Regions and Splits}:  
	Each node $v$ is associated with a set called a region, $A(v)$.
	The region of the root, $A(v_0)$ is given, while the regions of other nodes are defined recursively.
	Each internal node is associated with a split $c(v)$ which is a condition that partitions $A(v)$ into two complimentary sets. 
	One set is called the left set, $L(v)$, and the other is called the right set $R(v)$. 
	The left set is associated with the left daughter, and likewise the right set is associated with the right daughter.
	The region of node $v$ is obtained by applying the appropriate split condition to the region of its parent:
	$A(v)= L(p(v))$ if $v=l(p(v))$ or $A(v)=R(p(v))$ if $v=r(p(v))$.
	\item[] {\bf Values}: Every terminal node has a function, $f_v$, which maps from $A(v)$ to a set ${\mathcal R}_T$.
	These functions are collected into a set $F=\{f_v|\text{terminal nodes } v \in T\}$.
\end{itemize}

The regions of the terminal nodes of a binary tree with root $v_0$ are a partition of $A(v_0)$.
Every point in $x\in A(v_0)$ is contained in the region of exactly one terminal node $v$, and each terminal node $v$ is associated with a function $f_v$ which maps from $A(v)$ to $\mathcal{R}$, thus
a binary tree defines a function mapping from $A(v_0)$ to $\mathcal{R}$.
We use the symbol for a tree as the function it represents - that is given a tree $T$ with root $v_0$, the tree maps from a point $x \in A(v_0)$ to $T(x)$ or takes a subset $B \subseteq A(v_0)$ to its image in $\mathcal R$, $T(B)$.

Given a tree $T$ and a node $v$ in $T$, the subtree at $v$, $T_v$, is defined recursively as $v$ and all the nodes in $T$ which are daughters of nodes in the subtree at $v$.
A subtree $T_v$ defines a function which is the same function as $T$ but only defined on $A(v)$.

Several tree-based methods use models at terminal nodes. Bayesian CART \cite{Chipman2012} suggests using linear models at the terminal nodes, and Multivariate Adaptive Regression Splines use higher order models \cite{Friedman1991}.
However many methods, including CART, Random Forest, conditional trees and boosted trees, use constants for $f_v$. 

Note that the algorithm in Section \ref{sec:CombineRPF} is valid for any $c(v)$ which bipartitions $A(v)$ e.g. $L(v)$ and $R(v)$ could have non-linear boundaries. However, the complexity of determining if non-linear regions intersect can be very difficult. Many methods use conditions on one variables at a time.
These issues are discussed further in Section \ref{sec:specifications}.

There may be multiple trees which define the same recursive partition function.
For example the function in Figure \ref{fig:CombinedRecursivePartition} could be represented by the tree in Figure \ref{fig:CombinedRecursivePartition} or by the tree in Figure \ref{fig:CombinedRecursivePartitionAlternate}.
The correlation and distance we define in Section \ref{sec:distances} are based on the functions defined by trees and these measures do not account for discrepancies in the structures of trees.

\section{Combining trees}\label{sec:CombineRPF}

Recursive partitioning is typically applied to datasets with a univariate response $y$ and $p$ predictive variables $x_1,\ldots,x_p$, which could be a mixture of categorical and quantitative predictive variables.
The data structure we define and our algorithms for affine combinations and norms are generic, in the sense that they are independent from the types of the response and predictive variables.
These algorithms operate on any structure which can be represented by the tree defined in Section \ref{sec:tree} -- some examples are regression trees with multivariate response and splines.
The main restrictions required for sums and norms to be well defined are that the trees which are to be combined are contained in the same subspace of a vector space of functions, that is: (1) trees must have the same domain and range, (2) the range must be a normed vector space, and (3) a measure must be provided for the domain so that the mean and variance of a tree are well-defined.

As a precursor to computing affine combinations and norms of trees we give an algorithm
which takes two trees with the same domain, $\D$, such as $T^1$ and $T^2$ in Fig. \ref{fig:RecursivePartitionFunctionsExample} and combines them into a tree $\T$, and a set of vectors of functions indexed by the terminal nodes of $\T$, $\F=\{F_v| \forall \text{ terminal nodes } v \in \T\}$.
The vectors of functions at the terminal nodes of $\T$ map from $D$ to the product space $\R_{T^1} \times \R_{T^2}$.
A single tree which represents a combination of $T^1$ and $T^2$ must exist, and can be obtained by 
partitioning the regions of terminal nodes of $T^1$ into smaller regions using the splits from $T^2$.
More details about the method are presented in Sections \ref{sec:ExtractSubtree}  and \ref{sec:CombineRPFAlgorithm}.

A collection of trees, $T^1,\ldots,T^m$, mapping from the same domain $\mathcal D$, can be represented with a tree, $\T$, mapping from $\mathcal D$ to the product space $\prod_{j=1}^m {\mathcal R}_{T^j}$.
$\T$ may be obtained by iteratively applying the Tree Combiner Algorithm. 
Suppose that $T^1,\ldots,T^m$ map to the same vector space with real scalars.
Let $\alpha=(\alpha_1,\ldots,\alpha_m)$ be a vector in $\mathbb{R}^m$.
When the vector valued function $(T^1,\ldots,T^m)$ is represented by $\T$, then a tree representing the weighted sum $\alpha_1T^1+\ldots +\alpha_mT^m$ may be obtained from $\T$, by replacing the vector of functions $F_v=(f^1_v,\ldots,f^m_v)$ at each terminal node $v$ in $\T$ by the inner product of $F_v$ and $\alpha$.

The content of this Section is organized as follows.
Section \ref{sec:ExtractSubtree} describes the details of an important subroutine which collects the subtree of a recursive partition function restricted to a subset of its domain. 
Section \ref{sec:CombineRPFAlgorithm} describes the main algorithm, without assumptions about the type of domain, the types of partitions, or the types of functions at terminal nodes. 
Section \ref{sec:correctness} gives a proof of correctness and analysis of the computational cost of the algorithm.
Section \ref{sec:SplitRegion} describes details for how to check if a split intersects a region -- Section \ref{sec:SplitRegionDiscrete} concerns discrete sets, Section \ref{sec:SplitRegionUnivariate} deals with splits on one continuous variable, and Section \ref{sec:Preliminaries} describes methods for multivariate splits.

\subsection{Extracting a subtree in a region of the domain}\label{sec:ExtractSubtree}

To begin with a simpler question than finding a tree which represents two trees simultaneously, we ask, given a tree $T$, with root $v_0$ and $B \subset A(v_0)$, find a tree $T'$ such that $T'(B)=T(B)$, and $T'$ only uses conditions from nodes of $T$ with regions that intersect $B$.

Initially, $T'$ will be a single node $w_0$ with no children, and the region of $w_0$ is $B$.
To begin with we must find the minimum part of $T$ which contains $B$ and then extract the parts of that subtree which intersect $B$.
We define a recursive algorithm which operates on a terminal node of $T'$, $w$, and a node $v$ of $T$ such that $A(w) \subseteq A(v)$.
We can start at the roots $v=v_0$ and $w=w_0$.
If $A(w)$ intersects both $L(v)$ and $R(v)$, then we must use $c(v)$ to partition $A(w)$.
After splitting $w$ continue recursively with the left and right daughters of $w$ and $v$.
However, if $A(w)$ is a subset of either $L(v)$ or $R(v)$ then we must can move on to compare $w$ with $l(v)$ or $r(v)$, whichever the case may be.
Otherwise, if $v$ is terminal then let $f_w=f_v$.

Let $x$ be an element of $B$ and let $q$ be the node of terminal node of $T$ such that $x \in A(q)$.
Whenever the algorithm is called, the $A(w)$ is the intersection of $A(v)$ and $B$.
Therefore, the terminal node in $T'$ which contains $x$ will be associated with $f_q$.
Thus, for all $x \in B$, $T'(x)=T(x)$.

\begin{algorithm}[H]
	\KwData{$w$ is a terminal node in $T'$ \\ $v$ is a node in $T$ s.t. for any $q$ from $T$ if $A(q) \cap A(w) \neq \emptyset$ then $q$ is in the subtree at $v$}
	\KwResult{$T'_w$ is equivalent to $T_v$ over $A(w)$}
	\eIf{$A(w) \cap R(v) \neq \emptyset$ and $A(w) \cap L(v) \neq \emptyset$}{
		$c(w) \leftarrow c(v)$\;
		create $l(w)$ and $r(w)$ \;
        {\bf collect}$(l(w),l(v))$\;
        {\bf collect}$(r(w),r(v))$\;
	}{
	\If{$A(w)\cap L(v) \neq \emptyset$}{{\bf collect}$(w,l(v))$}
	\If{$A(w)\cap R(v) \neq \emptyset$}{{\bf collect}$(w,r(v))$}
	}
\caption{{\bf collect}$(w,v)$}\label{alg:ExtractBinT}
\end{algorithm}

\begin{figure}
	\includegraphics[width=9cm]{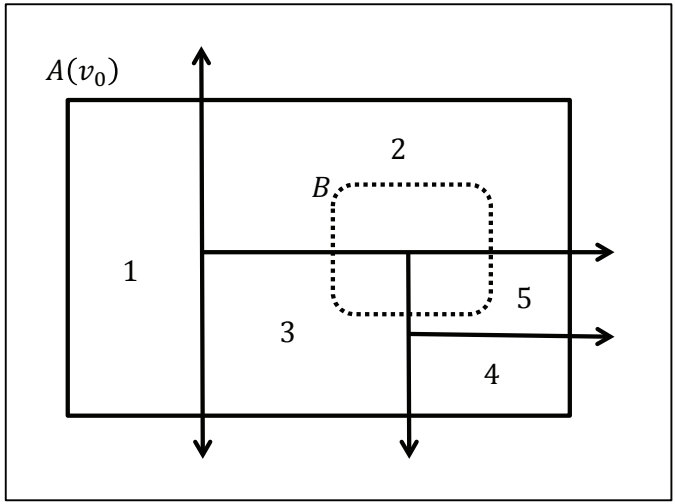}
	\includegraphics[width=9cm]{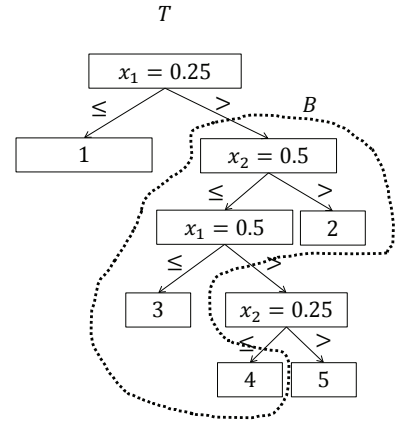}
	\includegraphics[width=18cm]{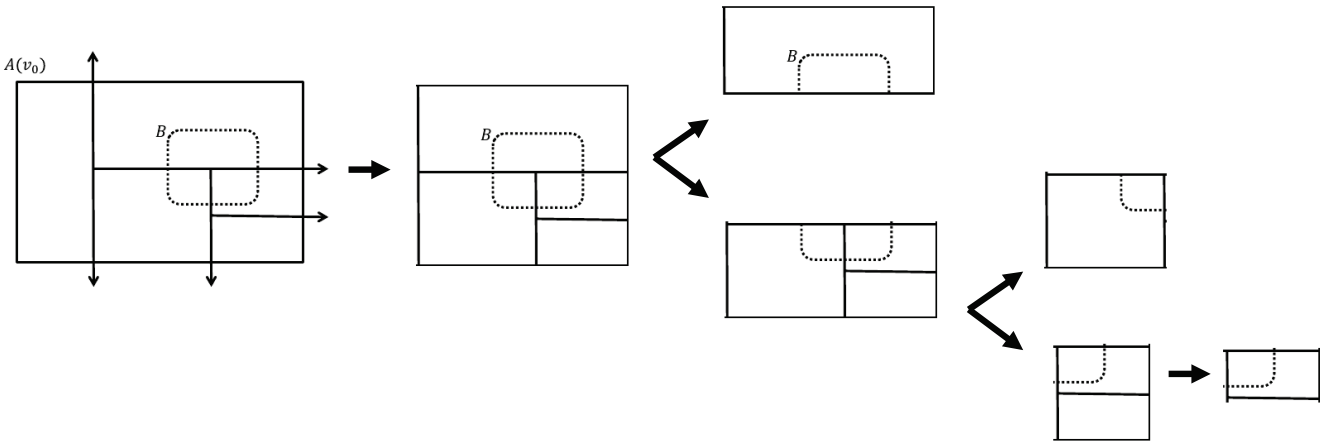}
	\caption{Illustration for extracting part of a tree over a subset $B$. Algorithm \ref{alg:ExtractBinT} locates the node with smallest region containing $A(w)$, and then splits $T'$ according to the condition in that node. Recursion continues, either moving to the left or right daughter in $T$ or partitioning until a terminal node of $T$ is reached.}
	\label{fig:collectBinTExample}
\end{figure}

\subsection{Tree Combiner Algorithm}\label{sec:CombineRPFAlgorithm}

The combined tree, $\T$, is initialized as a single node, $w_0$, with region $A(w_0)$ equal to $D$.
The main subroutine, ${\bf combine}$ is used recursively.
The inputs for ${\bf combine}$ are a node $u$ from $T^1$, a node $v$ from $T^2$, and a terminal node $w$ from $\T$, such that $A(w)$ is a subset of $A(u)$ and $A(v)$. 
When the recursion is over $\T_w$ and $\F_w$ are equivalent to $(T^1_u,T^2_v)$ over $A(w)$.
The algorithm reaches a base case when either $u$ or $v$ is a terminal node, where the problem of combining $T^1_u$ and $T^2_v$ reduces to the problem of extracting the part of a tree in a region. 
If neither $u$ nor $v$ is terminal and at least one of $c(u)$ and $c(v)$ can partition $A(w)$
then the algorithm partitions $A(w)$ and is called recursively on the daughters of $w$, and the appropriate nodes in $T^1$ and $T^2$.
There are four recursive cases: {\bf intersection absent}, {\bf crossing splits}, {\bf parallel splits}, and {\bf identical splits} which are defined in Section \ref{sec:case_descriptions}.

Details for how to proceed in each recursive case are described in Section \ref{sec:CasesForSplitPushNoIntersection}, \ref{sec:CasesForSplitPushCrossing}, \ref{sec:CasesForSplitPushParallel}, and \ref{sec:CasesForSplitPushIdentical}.
When Algorithm \ref{alg:TreeCombiner} reaches a terminal node of either $T^1$ or $T^2$, that branch of the recursion ends. 
Details for how to update $\T$ when a terminal node is reached are described in Section \ref{sec:CasesForSplitPushTerminal}.\\

\begin{algorithm}[H]
	\KwData{$w$ is a terminal node in $\T$ \\ $u$ is a node in $T^1$ s.t. for any $q$ from $T^1$ if $A(q) \cap A(w) \neq \emptyset$ then $q$ is in the subtree at $u$\\ $v$ is a node in $T^2$ s.t. for any $q$ from $T^2$ if $A(q) \cap A(w) \neq \emptyset$ then $q$ is in the subtree at $v$\\}
	\KwResult{$\T_w$ and $\F_w$ which are equivalent to $(T^1_u,T^2_v)$ over $A(w)$}
	\eIf{$u$ or $v$ is terminal}{go to {\bf terminal node reached}, Section \ref{sec:CasesForSplitPushTerminal}}
	{
		\If{at least one of $c(u)$ or $c(v)$ do not intersect}{ $A(w)$ go to {\bf intersection absent}, Section \ref{sec:CasesForSplitPushNoIntersection}}
		\If{if $c(u)$ and $c(v)$ cross in $A(w)$ }{go to {\bf crossing splits}, Section \ref{sec:CasesForSplitPushCrossing}}
		\If{if $c(u)$ and $c(v)$ are parallel in $A(w)$}{ go to {\bf parallel splits}, Section \ref{sec:CasesForSplitPushParallel}}
		\If{if $c(u)$ and $c(v)$ are identical in $A(w)$ }{go to {\bf identical splits}, Section \ref{sec:CasesForSplitPushIdentical}}
	}
\caption{{\bf combine}$(u,v,w)$}\label{alg:TreeCombiner}
\end{algorithm}


\subsubsection{Case descriptions}\label{sec:case_descriptions}
Before we discuss the recursive cases of Algorithm \ref{alg:TreeCombiner} in more detail it is helpful to consider discuss the conditions for the recursive cases and the tasks achieved by the main subroutines. 
The following two subroutines check if parts of trees meet certain conditions.
\begin{itemize}
	\item[] {\bf Check if a region can be partitioned by a condition:} Given two nodes $u$ and $v$ check if the condition $c(v)$ partitions $A(u)$, that is check if both the sets $A(u)\cap L(v)$ and $A(u) \cap R(v)$ are non-empty.
	\item[] {\bf Check if two splits are crossing, parallel or identical:} Given two nodes $u$ and $v$ determine which of the following disjoint and exhaustive events is true:
	\begin{itemize}
		\item[(i)] (crossing) none of the following sets are empty: $R(u)\cap R(v)$, $R(u)\cap L(v)$, $L(u)\cap R(v)$, and $L(u)\cap L(v)$
		\item[(ii)] (parallel) $L(u)\cap R(v)$ is empty and none of the following sets are empty:  $R(u)\cap L(v)$, $R(u)\cap R(v)$, and $L(u)\cap L(v)$
		\item[(iii)] (parallel) $L(v)\cap R(u)$ is empty and none of the following sets are empty:  $R(u)\cap R(v)$, $L(u)\cap R(v)$, and $L(u)\cap L(v)$ 
		\item [(iv)] (identicial) $c(u)$ and $c(v)$ are the same
	\end{itemize}
	When a region $A$ is reduced by both $c(u)$ and $c(v)$ it will be of interest to focus on whether these splits are crossing, parallel or identical inside of $A$. We can check if two splits are crossing, parallel or identical inside of $A$, by including intersection with $A$ in all of the above events.
	For example, $c(u)$ and $c(v)$ are crossing in $A$ if none of the following sets are empty: $R(u)\cap R(v) \cap A$, $R(u)\cap L(v) \cap A$, $L(u)\cap R(v) \cap A$, and $L(u)\cap L(v) \cap A$.
	\item[] {\bf Find subtreee in region:} Given a region $A$ and a tree $T$ find a tree $T'$ which is equivalent to $T$ over $A$ - see Section \ref{sec:ExtractSubtree} for details.
\end{itemize}

\subsubsection[Case: No Intersection]{Cases for Algorithm \ref{alg:TreeCombiner}: Intersection Absent}\label{sec:CasesForSplitPushNoIntersection}
There are several cases for how at least one of the conditions $c(u)$ and $c(v)$ does not intersect $A(w)$.
Suppose $c(u)$ does not intersect $A(w)$. 
Then $A(w)$ is contained either in the left piece at $u$, $L(u)$, or the right piece at $u$, $R(u)$.
If $A(w)$ is contained in $R(u)$ then it is necessary to explore the subtree at the right daughter of $u$, $r(u)$.
Algorithm \ref{alg:TreeCombiner} checks if any of its split cross over $A(w)$ and check how these pieces may interact with pieces in the subtree of $V$ at $v$.
Therefore, in that case, split and push is called for $r(u)$, $v$ and $w$.
All the possible cases and how to proceed in each case are outlined here.\\
Exactly one of the following is true: 
\begin{itemize}
	\item[(u,i)] the region of $w$ is in the left region of $u$, $A(w) \subseteq L(u)$
	\item[(u,ii)] the region of $w$ is in the right region of $u$, $A(w) \subseteq R(u)$
	\item[(u,iii)] $c(u)$ splits the region of $w$, $A(w)\cap c(u)\neq \emptyset$
\end{itemize}
and, exactly one of the following is true: 
\begin{itemize}
	\item[(v,i)] the region of $w$ is in the left region of $v$, $A(w) \subseteq L(v)$
	\item[(v,ii)] the region of $w$ is in the right region of $v$, $A(w) \subseteq R(v)$
	\item[(v,iii)] $c(v)$ splits the region of $w$, $A(w)\cap c(v)\neq \emptyset$
\end{itemize}

\begin{itemize}
	\item[] \bif $A(w) \subseteq L(u)$ and $A(w) \subseteq L(v)$ \bthen {\bf combine}$(l(u),l(v),w)$ 
	\item[] \bif $A(w) \subseteq L(u)$ and $A(w) \subseteq R(v)$ \bthen {\bf combine}$(l(u),r(v),w)$
	\item[] \bif $A(w) \subseteq L(u)$ and $A(w)\cap c(v) \neq \emptyset$ \bthen {\bf combine}$(l(u),v,w)$
	\item[] \bif $A(w) \subseteq R(u)$ and $A(w) \subseteq L(v)$ \bthen {\bf combine}$(r(u),l(v),w)$
	\item[] \bif $A(w) \subseteq R(u)$ and $A(w) \subseteq R(v)$ \bthen {\bf combine}$(r(u),r(v),w)$
	\item[] \bif $A(w) \subseteq R(u)$ and $A(w)\cap c(v) \neq \emptyset$ \bthen {\bf combine}$(r(u),v,w)$
	\item[] \bif $A(w)\cap c(u) \neq \emptyset$ and $A(w) \subseteq L(v)$ \bthen {\bf combine}$(u,l(v),w)$
	\item[] \bif $A(w)\cap c(u) \neq \emptyset$ and $A(w) \subseteq R(v)$ \bthen {\bf combine}$(u,r(v),w)$
\end{itemize}

\vspace{0.5cm}
\subsubsection[Case: crossing]{ Cases for Algorithm \ref{alg:TreeCombiner}: Crossing splits}\label{sec:CasesForSplitPushCrossing}
When $c(u)$ and $c(v)$ split $A(w)$ into four non-empty subsets, these splits are said to cross inside $A(w)$.
Only one of $c(u)$ and $c(v)$ can be used as the split for $c(w)$. 
Which of $c(u)$ and $c(v)$ is chosen is arbitrary, but this choice impacts which parts of $T^1$ and $T^2$ are used in the recursive calls to Algorithm \ref{alg:TreeCombiner}.
For example, if $c(u)$ is chosen as the split for $w$, then there is a recursive call for $l(u)$, $v$ and $l(w)$, and a recursive call for $r(u)$, $v$ and $r(w)$.
Pseudo-code for using $c(u)$ is given below, but for brevity, pseudo-code for using $c(v)$, which is analogous to the code for using $c(u)$, is omitted.\\ 

\noindent choose either $u$ or $v$ to split $A(w)$\\
suppose $u$ is chosen to split $A(w)$ then do the following
\begin{itemize}
	\item[] create daughters for $w$, $l(w)$ and $r(w)$
	\item[] let $c(w)=c(u)$ (thus $A(l(w))=A(w)\cap L(u)$ and $A(r(w))=A(w)\cap R(u)$)
	\item[] {\bf combine}$(l(u),v,l(w))$
	\item[] {\bf combine}$(r(u),v,r(w))$ 
\end{itemize}
if $v$ is chosen to split $w$, then do the above, but swap the roles of $u$ and $v$

\vspace{0.5cm}
\subsubsection[Case: parallel]{Cases for Algorithm \ref{alg:TreeCombiner}: Parallel Splits}\label{sec:CasesForSplitPushParallel}
When $c(u)$ and $c(v)$ are parallel in $A(w)$, it is possible to use either one as the split for $w$.
Since $c(u)$ and $c(v)$ are parallel in $A(w)$, that is the subsets they create are nested, when one is used for the split for $w$, the other is present in the region of just one of the daughters of $w$.
For example if $c(u)$ is used as the split for $w$, and $R(v)\cap A(w) \subset R(u)\cap A(w)$ then $c(v)$ intersects $A(r(w))$ but not $A(l(w))$.
Therefore when Algorithm \ref{alg:TreeCombiner} is called on $r(u)$, $v$, and $r(w)$, and called on $l(u)$, $l(v)$, and $l(w)$.
The cases when $c(u)$ is used to split $A(w)$ are outlined below, and since they are similar, the instructions for when $c(v)$ is used to split $A(w)$ are omitted. \\

\noindent choose either $u$ ro $v$ to split the region of $w$\\
suppose $u$ is chosen to split $w$ then do the following\\
create daughters for $w$, $l(w)$ and $r(w)$\\
let $c(w)=c(u)$ (thus $A(l(w))=A(w)\cap L(u)$ and $A(r(w))=A(w)\cap R(u)$)\\
There are two cases for recursion:
\begin{itemize}
	\item[(i)] \bif $c(v)$ intersects $L(u)\cap A(w)$ \bthen 
	\begin{itemize}
		\item[] {\bf combine}$(l(u),v,l(w))$
		\item[] {\bf combine}$(r(u),r(v),r(w))$
	\end{itemize}
	\item[(ii)] \bif $c(v)$ intersects $R(u)\cap A(w)$ \bthen
	\begin{itemize}
		\item[] {\bf combine}$(l(u),l(v),l(w))$
		\item[] {\bf combine}$(r(u),v,l(w))$ 
	\end{itemize}
\end{itemize}
\bif $v$ is chosen to split $w$ \bthen do the above, but swap the roles of $u$ and $v$

\subsubsection[Case: identical]{Cases for Algorithm \ref{alg:TreeCombiner}: Identical Splits}\label{sec:CasesForSplitPushIdentical}
When $c(u)$ and $c(v)$ induce the same partition on $A(w)$ either one can be used as the split at $w$. Pseudo-code for this situation is given below: \\
let $c(w)=c(u)$\\
create daughters for $w$, $l(w)$ and $r(w)$\\
\bif $R(u)\cap A(w)=R(v)\cap A(w)$ \bthen
\begin{itemize}
	\item[] {\bf combine}$(l(u),l(v),l(w))$ 
	\item[] {\bf combine}$(r(u),r(v),r(w))$
\end{itemize}
\belse ($R(u)\cap A(w)=L(v)\cap A(w)$)
\bif $R(u)\cap A(w)=R(v)\cap A(w)$ \bthen
\begin{itemize}
	\item[] {\bf combine}$(l(u),r(v),l(w))$
	\item[] {\bf combine}$(r(u),l(v),r(w))$
\end{itemize}

\subsubsection[Case: Terminal Nodes]{Cases for Algorithm \ref{alg:TreeCombiner}: Terminal Node Reached}\label{sec:CasesForSplitPushTerminal}
When either node $u$ or node $v$ is terminal a base case is reached.
There are three possibilities: both $u$ and $v$ are terminal, only $u$ is terminal, or only $v$ is terminal. 
If both $u$ and $v$ are terminal, then node $w$ is assigned their values, $F_w=(f_u,f_v)$.
If only $u$ is terminal, then the region of $w$, $A(w)$ is further partitioned by $T^2$. 
Therefore, it is necessary to collect the subtree of $T^2$, contained in $A(w)$, denoted $T^2_{A(w)}$.
Details for how to obtain $T_{A(w)}$ are in Section \ref{sec:ExtractSubtree}.
Once $T_{A(w)}$ is obtained, the values in its terminal nodes are combined with $f_u$. 
The resulting tree is appended to $T$ at $w$.
When only $v$ is terminal, the operations are similar, only the roles of $u$ and $v$, and $T^1$ and $T^2$ are switched.  
Pseudo-code is given below:
\begin{itemize}
	\item[] \bif $u$ is terminal \bthen
	\begin{itemize}
		\item[] copy the subtree at $v$ inside the region $A(w)$, call it $T_{A(w)}$
		\item[] for each terminal node $v'$ in $T_{A(w)}$ replace its value, $f_{v'}$ with $F_{v'}=(f_u,f_{v
			'})$
	\end{itemize}
	\item[] \belse \bif $v$ is terminal \bthen do the same, but swap the roles of $u$ and $v$
\end{itemize}

%
%

\subsection{Correctness and Computational Cost}\label{sec:correctness}
The proof of correctness is much simpler in the special case when we assume that
whenever given the option to choose a split for $w$ from either $u$ or $v$, $c(u)$ is always chosen.
Assuming $c(u)$ is chosen whenever possible, Algorithm \ref{alg:TreeCombiner} performs a depth first search of $T^1$, until a terminal node of $T^1$ is reached.
In some calls to Algorithm \ref{alg:TreeCombiner}, the split from $v$ will not intersect $A(u)$.
In this case proceeds by a recursive call to Algorithm \ref{alg:TreeCombiner} for $u$, and for whichever daughter of $v$, either $l(v)$ or $r(v)$, has a region, $L(v)$ or $R(v)$ which contains $A(u)$.
Since $T^2$ is a tree, any splits from $T^2$ which intersect $A(u)$ must be present in the subtree of that daughter.
This guarantees that whenever $u$ is a terminal node of $T^2$ any portion of $T^2$ that intersects $A(u)$ will be contained in the subtree of $T^2$ at $v$.
Hence, once a terminal node of $T^1$ is reached, the subtree at $v$ contains the entire subtree of $T^2$ inside $A(u)$.

In general, the Algorithm \ref{alg:TreeCombiner} is valid whether $c(u)$ or $c(v)$ is used when given the choice to split $w$.

\begin{lem}\label{lem:SubtreesValid}
	At each call to Algorithm \ref{alg:TreeCombiner} the subtrees at the nodes of $u$ and $v$, contain any parts of $T^1$ and $T^2$ which intersect with $A(w)$.
\end{lem}
\begin{proof}
	For the first call to Algorithm \ref{alg:TreeCombiner}, $u$ and $v$ are the roots of $T^1$ and $T^2$, respectively; and $A(w)$ is the entire domain $D$.
	We will argue by induction.
	We must prove that the inductive hypothesis is true for the input to Algorithm \ref{alg:TreeCombiner}, then it must be true for recursive calls to Algorithm \ref{alg:TreeCombiner}.
	This can easily be verified the four recursive cases as follows:
	
	\begin{itemize}
		\item[(i)] (intersection absent) If neither $c(u)$ nor $c(v)$ intersects $A(w)$, then $A(w)$ is completely contained inside one daughter of $u$ and one daughter of $v$. The algorithm locates the daughters which contain $A(w)$ can calls Algorithm \ref{alg:TreeCombiner} on that case. 
		Suppose just one of $c(u)$ doesn't intersect $A(w)$ but $c(v)$ does, then the algorithm identifies which daughter of $c(u)$ contains $A(w)$, and calls split and push on that daughter, $v$, and $w$. 
		The last case is symmetric.
		\item[(ii)] (crossing splits) In this case $c(u)$ and $c(v)$ divide $A(w)$ into four non-empty subsets.
		Suppose $c(u)$ is used for the split of $A(w)$.
		Since we have assumed the inductive hypothesis, any subtrees of $T^1$ that intersect $A(w)$ are in the subtree at $u$.
		Since the subtrees at $l(u)$ and $r(u)$ are restricted to $L(u)$ and $R(u)$, respectively, the only subtrees of $T^1$ which intersect with $L(w)$ and $R(w)$ in the subtree at $l(u)$ and $r(u)$, respectively. 
		Since $c(v)$ crosses $c(w)$, subtrees in $T^2$ which intersect $A(l(w))$ and $A(r(w))$ are in the subtree at $v$. 
		Hence Algorithm \ref{alg:TreeCombiner} is called $l(u)$, $v$ and $l(w)$, and $r(u)$, $v$ and $r(w)$.
		The symmetric case, i.e. using $c(v)$ as the split for $w$, is similar. 
		\item[(iii)] (parallel splits) In this case $c(u)$ and $c(v)$ divide $A(w)$ into three non-empty subsets. 
		Suppose $c(u)$ is used as the split for $A(w)$.
		The split $c(v)$ intersects with the region for one daughter of $w$ and not the other.
		Suppose $c(v)$ intersect with the region of the left daughter of $w$, $A(l(w))$.
		In this case an subtrees of $T^1$ and $T^2$ which intersect with $A(l(w))$ are in subtrees at $l(u)$ and $v$, respectively. 
		Since $c(u)$ and $c(v)$ are parallel and $c(v)$ intersects $A(l(w))$, the region of the right daughter of $w$, $A(r(w))$, is completely contained in either $A(l(v))$ or $A(r(v))$.
		Suppose $A(r(w)) \subset A(r(v))$.
		Thus, subtrees of $T^1$ and $T^2$ which intersect $A(r(w))$ are in subtrees at $r(u)$ and $r(v)$.
		Symmetric cases, with the roles of $u$ and $v$, and/or the roles of their left and right daughters switched, follow similar lines of reasoning.
		\item[(iv)] (identical splits) Splits $c(u)$ and $c(v)$ are the same, possibly swapping left and right partitions. 
		Suppose that $A(w) \cap L(u) = A(w) \cap L(v)$. 
		Any subtrees of $T^1$ and $T^2$ which intersect $A(l(w))$ are contained in the subtrees at $l(u)$ and $l(v)$; and likewise, any subtrees of $T^1$ and $T^2$ which intersect $A(r(w))$ are contained in the subtrees at $r(u)$ and $r(v)$, respectively. 
	\end{itemize}
\end{proof}

%

Assuming $u$ is chosen whenever the choice between $u$ and $v$ must be made, Algorithm \ref{alg:TreeCombiner} is called once for each node in $T^1$.
Once a terminal node of $T^1$ is reached, then Algorithm \ref{alg:ExtractBinT} is called at most once for each node in $T^2$. Therefore, in the worst case the total number of calls to Algorithm \ref{alg:TreeCombiner}, and Collect Subtree, is $n_1n_2$, where $n_1$ and $n_2$ are the number of nodes in $T^1$ and $T^2$ respectively.

\section{Tree distances and correlations}\label{sec:distances}
The goal of this section is to describe, distances and correlations for trees, which quantify the degree of difference between two trees.
We also describe how to efficiently compute distance and correlation between trees as an extension of Algorithm \ref{alg:TreeCombiner}.

\subsection{Tree Distances}\label{sec:TreeDistances}
The norm of a function, $f$, with respect to a measure, $p$, on its domain, $\D$, is
\begin{equation}
\norm{f}=\left(\int_\D f(x)^2dp(x)\right)^{1/2}.
\end{equation}
The norm of the difference between two functions $f$ and $g$, defines a metric, also called a distance, 
\begin{equation}
\norm{f-g}=\left(\int_\D (f(x)-g(x))^2dp(x)\right)^{1/2}.
\end{equation}

For trees, the square of the norm can be decomposed into a sum of the squares of norms of the set of terminal nodes, $W$,
\begin{equation}
\norm{T}^2=\sum_{w \in W}\int_{A(w)} \norm{f_w(x)}^2dp(x),
\end{equation}
since the regions of the terminal nodes are a partition of the domain. 
The sqaure of the distance between two trees, $\norm{T^1-T^2}^2$, can be computed by: (1) combining them into a single tree $\T$ with each terminal node $w\in W_\T$ associated with a multifunction $(f_w^1,f_w^2)$, as described in Section \ref{sec:CombineRPF}, and (2) computing the sum of the sqaure of the distance between the functions at each terminal node of $\T$, $W_\T$, 
\begin{equation}
\norm{T^1-T^2}^2=\sum_{w \in W_\T}\int_{A(w)} \norm{f^1_w(x)-f^2_w(x)}^2dp(x).
\end{equation}
For regression problems with continuous response, it is common to use a single scalar value at each terminal node, $f_w(x)=a_w$, and in this case the distance between two trees simplifies to 
\begin{equation}\label{eq:TreeDistRegScalar}
\norm{T^1-T^2}^2=\sum_{w \in W_\T}(a^1_w-a^2_w)^2\int_{A(w)}1 dp(x).
\end{equation}
However, a different formula is required for classification and density estimation since in this context the trees map to sets of classes, or assignments of probabilities to sets of classes, and typically a metric to quantify the difference between classes is not provided.
Classification trees often provide estimates of the class probabilities. 
Treating estimates of class probabilities as vectors, we can quantify the difference between two estimates of class probabilities as the norm of their difference.
For classification trees which do not provide estimates of class probabilities a simple solution is to use a probability of 1 for the predicted class.
Consider a classification problem with $S$ classes.
We assume that classification tree, $T^i$, at terminal node $w$ maps every point $x \in A(w)$, to a vector of class probabilities $\left[f^i_{ws}|s\in S\right ] \in \mathbb{R}^{|S|}$.
Let $T^1$ and $T^2$ be classification or density estimation trees.
When $T^1$ and $T^2$ are represented by a single tree $\T$, the values at each terminal node $w \in W_\T$ is a $|S|$ by $2$ dimensional matrix $\left [ f^1_{ws},\;f^2_{ws} |s\in S \right ]$.
Thus the distance between $T^1$ and $T^2$ is
\begin{equation}\label{eq:TreeDistClass}
\norm{T^1-T^2}^2=\sum_{w \in W_\T}\sum_{s\in S}{(f^1_{ws}-f^2_{ws})^2}\int_{A(w)} 1 dp(x).
\end{equation}

Equations \ref{eq:TreeDistRegScalar} and \ref{eq:TreeDistClass}, both depend on the measure $p$ and the regions of terminal nodes, $A(w), \; w \in W_\T$. 
Ideally, the measure $p$ should reflect the unknown density from which the sample is obtained. 
Since the distribution is unknown we will have to estimate it and/or make assumptions. For instance we could assume that the distribution of the data comes from a uniform distribution on $\D$, and use the uniform measure when computing the weight of each terminal node $w \in W_T$.
If data is available when computing the distance between $T^1$ and $T^2$ then we can use the proportion of the sample in the region of a terminal node as the weight for that node. 
This choice of measure would cause the distance to capture the discrepancy between $T^1$ and $T^2$ in regions which support the majority of the mass of the observed distribution, while ignoring their difference in regions with no data.


If we use a uniform density for $p$ then the distance between $T^1$ and $T^2$ can be computed with a recursive algorithm, which is more efficient than computing each term of the sum in Equation \ref{eq:TreeDistRegScalar} or Equation \ref{eq:TreeDistClass} independently. We use $p(A)$ to denote the measure of region $A$, $p(A)=\int_A 1 dp(x)$.

\begin{algorithm}[H]
	\KwData{a node $w$ in the combined tree $\T$ from trees $T^1$ and $T^2$}
	\KwResult{$\norm{T^1-T^2}^2$ over $A(w)$ normalized by the measure of $A(w)$}
	\eIf{$w$ is terminal}{
		{\bf return } $\norm{f^1_w-f^2_w}^2$
	}{
		$p_r\leftarrow p(A(r))/p(A(w))$\\
		$p_l\leftarrow p(A(l))/p(A(w))$\\
		{\bf return}  $p_l${\bf tree\_dist}$(l)+p_r${\bf tree\_dist}$(r)$
	}
\caption{{\bf tree\_dist}$(w)$}\label{alg:TreeDistRecursiveVol}
\end{algorithm}


Generalizations of distances to cases when response variables are elements of metric spaces other than $\mathbb{R}$, e.g. $\mathbb{R}^p$, would require a different bifunction to measure the difference between $T^1$ and $T^2$, but nevertheless methods for computing such distances would follow the same two steps as the univariate response case: (1) combine the trees and (2) reduce the problem a sum over the terminal nodes of the combined tree.

We created a sample of 100 multivariate predictor and univariate response pairs, $(x^1,y^1), \ldots,(x^{100},y^{100})$, sampled with $x^i$ uniformly random in the region $[0,10]\times [0,10]$ and $y^i$ is obtained by evaluating the tree in Figure \ref{fig:SimpleTreeModel} at $x^i$ ($y^i$ is not corrupted with noise).
We used the Random Forest R Package to generate an ensemble of trees from this data and computed the distance between each pair of trees.
Multidimensional scaling plots of these tree distances are in Figure \ref{fig:MDSCloud}.
Non-linear structures are apparent in these projections of the ensemble into three dimensional Euclidean space. 
Such patterns suggest motifs or families of trees in the ensemble.
This result indicates that this metric could be used in a method for selecting a subset of representative element from the ensemble.
However, the development of a formal method is left as a topic for further research.

\begin{figure}
	\centering
	\includegraphics[width=3in]{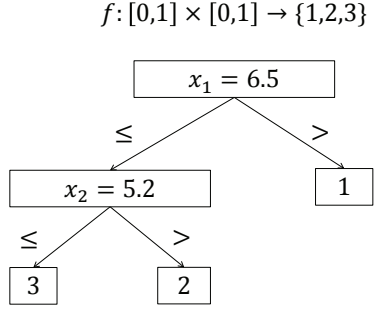}
	\caption{Tree model for computational experiments.}
	\label{fig:SimpleTreeModel}
\end{figure}

\begin{figure}
	\centering
	\includegraphics[width=3in]{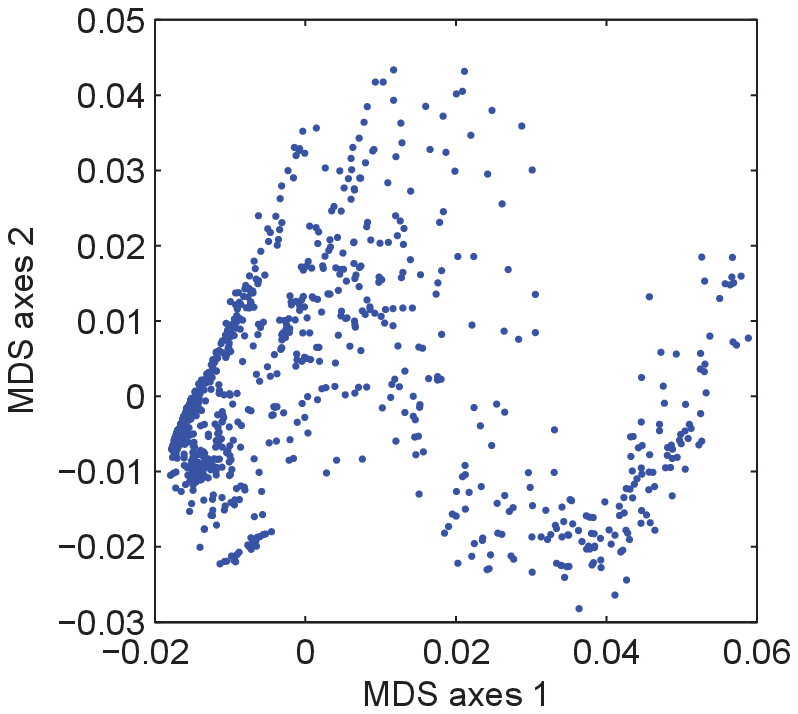}
	\includegraphics[width=3in]{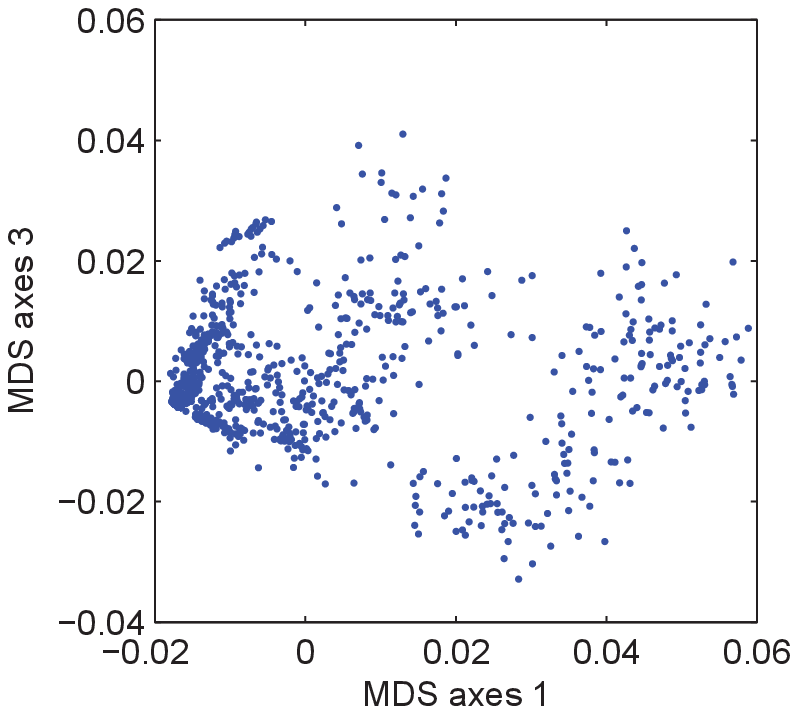}
	\caption{Multidimensional scaling plots of point cloud of a sample of trees from Random Forest, where distances between trees are measured in ${\mathcal{L}_2}$ function space. 
	}
	\label{fig:MDSCloud}
\end{figure}

\subsection{Tree Correlations}\label{sec:TreeCorrelations}
The distance between trees will quantify their difference, however, it is not standardized relative to the norms of the functions.
Correlation is an alternative quantification of the similarity between trees which is on a standardized scale between -1 and 1.
In this section we define correlation between two trees as a generalization of the commonly used Pearson correlation for random variables.

The correlation between two trees, $T^1$ and $T^2$, is their covariance standardized by the product of their standard deviations,
\begin{equation}\label{eq:RegTreeCorr}
\rho_{T^1,T^2} =\frac{\sigma_{T^1,T^2}}{\sigma_{T^1}\sigma_{T^2}}.
\end{equation}
The covariance between regression trees $T^1$ and $T^2$ quantifies their similarity as the integral of the product of their deviance from their respective mean values at each point $x$ in their domain with respect to a measure $p(x)$,
\begin{equation}\label{eq:RegTreeCov}
\sigma^2_{T^1,T^2} =\int_D (T^1(x)-\mu_{T^1})(T^2(x)-\mu_{T^2})dp(x),
\end{equation}
where the mean and standard deviation of regression tree $T^i$ are
\begin{equation}\label{eq:RegTreeMean}
\mu_{T^i} = \int_{D} T^i(x) d p(x)
\end{equation}
and
\begin{equation}\label{eq:RegTreeVar}
\sigma^2_{T^i} =\int_{D} (T^i(x)-\mu_{T^i})^2 dp(x).
\end{equation}

Separating these integrals over the disjoint regions of the terminal nodes the mean and standard deviation of regression tree $T^i$ can be expressed as sums over $W_{T^i}$, the set of terminal nodes,
\begin{equation}
\mu_{T^i} =\sum_{w \in W_{T^i}} \int_{A(w)} T^i(x) d p(x)
\end{equation}
and
\begin{equation}
\sigma^2_{T^i} =\sum_{w \in W_{T^i}} \int_{A(w)} (T^i(x)-\mu_{T^i})^2 dp(x).
\end{equation}
Similarly, the covariance for $T^1$ and $T^2$ can be expressed as a sum
over $W$, the set of nodes in a combined tree representing $T^1$ and $T^2$,
\begin{equation}
\sigma^2_{T^1,T^2} =\sum_{w\in W}\int_{A(w)} (T^1(x)-\mu_{T^1})(T^2(x)-\mu_{T^2})dp(x).
\end{equation}
Recursive algorithms with the same pattern and assumptions as Algorithm \ref{alg:TreeDistRecursiveVol} can be formulated to compute tree means, variances, and covariances.

Regarding classification problems, as discussed in Section \ref{sec:TreeDistances}, a rule for quantifying the discrepancy between the classes is not always available, however we can use the norm of the difference between the probability estimates of different classes to quantify the discrepancy between class predictions.
If the response of classification trees is a vector of class probabilities, the definitions of correlation, covariance, mean, and variance for regression trees (\ref{eq:RegTreeCorr}-\ref{eq:RegTreeVar}) no longer apply.
However, generalizations of these concepts can be defined.

The variance covariance matrix for a probability density tree quantifies the degree to which the probability of classes vary together, either above or below the average class probabilities.
This can be used to diagnose the extent to which it is hard to discriminate between two classes.
 
\subsection{Distances between Forests}
Consider two forests $f^1,\ldots,f^J$ and $g^1,\ldots,g^K$, where each tree maps from the same domain to the real numbers, and their aggregate functions $F(x)=\sum_{j=1}^J f^j(x)$ and $G(x)=g^1(x)+\ldots+g^K(x)$.
The squared of the 2-norm or squared-distance between $F$ and $G$ with respect to a measure $p$
is 
\begin{eqnarray}\label{eq:ForestDistance}
d(F,G)=&\int_D (F(x)-G(x))^2dp(x)\\
\end{eqnarray}
When the measure $p$ is restricted to a finite set of points masses, not too large in number, it will be possible to compute this distance directly from the representation of $F$ and $G$ as sums of trees.
However, when $p$ is continuous, or the if the $p$ is constituted by a vast number of discrete points, it is not possible to compute the value of the $d(F,G)$ directly by formula \ref{eq:ForestDistance}.
If $J$ and $K$ are not too large, and the dimension of the domain of $F$ and $G$ is not too large, then it may be possible to represent $F-G$ as a single tree using Alg. \ref{alg:TreeCombiner}, evaluate the distance using Alg. \ref{alg:TreeDistRecursiveVol}.
However, since the size of the combined tree will grow multiplicatively due to the intersection of splits from the different trees the size of the combined tree $F-G$ will be much larger than the sum of the sizes of the individual trees $f^1,\ldots,f^J$ and $g^1,\ldots,g^K$.
With simplifying assumptions we can show that the size of the combined tree could grow exponentially in
the number of trees a forest.
Suppose each tree partitions the domain $D=[0,1]^M$ into two pieces by partitioning on dimension $m$.
Let $(x_1,\ldots,x_M)$ be a point in $D$. Suppose an ensemble is composed of $K$ trees representing functions of the form $f^k(x)=a_k$ if $x_{m_k}\leq c_k$ and $f(x)=b_k$ if $x_{m_k}>c_k$ for $k=1,\ldots,K$. 
Suppose that the first $k_1$ trees split on dimension $1$, that is $m_k=1$, and for each $m=2,...,M$, the next $k_m$ trees split on dimension $m$.
The sum of the first $k$ functions is $F_k=\sum_{i=1}^k f^k$.
How many rectangular cells does the function $F_K$ partition $D$ into?
Assuming non-degeneracy $F_{k_1}$ has $k_1$ splits on dimension $1$, and thus partitions the $D$ into 
$k_1+1$ cells.
The plane $x=c_{k_1+1}$ intersects all the planes $x=c_1,\ldots,x=c_{k_1}$, and thus $F_{k_1+1}$ partitions $D$ into $2(k_1+1)$ cells. Since $c_{k_1+2}\neq c_{k_1+1}$, adding $f^{k_1+2}$ introduces another $k_1$ cells, therefore $F_{k_1+2}$ partitions $D$ into $3(k_1+1)$ cells.
Following the same argument $F_{m_1+m_2}$ partitions $D$ into $(k_1+1)(k_2+1)$ cells.
Continuing the same argument for $m=3,\ldots,M$, we find that $F_M$ partitions $D$ into $n_M=(k_1+1)\times \ldots \times (k_p+1)$ cells. 
Representing a partition of $D$ into this many cells requires a binary tree with $n_M$ of leaf nodes, and $n_M-1$ internal nodes.
So for example with just one tree per dimension, $m=1,\ldots,M$, the representative CART would have $2^{M+1}-1$ nodes.
In stark contrast the total number of nodes in all trees of such a forest is $3M$.

Expanding the squared difference $(F(x)-G(x))^2$ and using the linearity of the integral operator the squared distance between $F$ and $G$ can be computed as sums of much simpler terms,
\begin{eqnarray}
d(F,G)=&\int_D (\sum_j (f^j(x))^2+\sum_k (g^k(x))^2 - 2\sum_{j,k}f^j(x)g^k(x) )dp(x)\\
=& (\sum_j \int_D(f^j(x))^2+\sum_k \int_D(g^k(x))^2 - 2\sum_{j,k}\int_Df^j(x)g^k(x) )dp(x).
\end{eqnarray}
The inner product of two trees $\int_D T^1(x)T^2(x)$ can be computed using an algorithm with the same data and recursive format as Alg. \ref{alg:TreeDistRecursiveVol}, and for the base case, when $w$ is a terminal node, the algorithm will return $f^1_w \times f^2_w$ instead of $\norm{f^1_w-f^2_w}$.

\section{Solutions for subproblems in specific tree contexts}\label{sec:specifications}
\subsection{Checking if a split divides a region}\label{sec:SplitRegion}

\subsubsection{Univariate splits for discrete variable}\label{sec:SplitRegionDiscrete}
Suppose split $c$ acts on a discrete variable $x$. Then the condition $c$ is intersect a subset of the domain $A$, if it divides the elements of $x$ in $A$ into two non-empty sets.

\subsubsection{Univariate splits for continuous variables}\label{sec:SplitRegionUnivariate}
When splits are made on a single variable at a time the intersection of a split and a region can be achieved
by testing the intersection of the split and the restriction of the region to the same variable.
That is if the region is defined by univariate linear inequalities, then only the inequalities involving the variable for the split being tested are relevant.
Likewise, for categorical variables it would only be necessary to test subsets of the variable for the split being tested.

\subsection{Multivariate splits for continuous variables}

The purpose of this section is to describe the geometry of recursive partitions when multivariate splits are used for continuous variables, which results in polyhedral regions.
Generally, computing volumes of polyhedra or integrals of functions over polyhedra requires exponential time algorithms, or randomized approximations are used. 
Hence computing the $\mathcal{L}_2$ distance between recursive partition function with multi-variate splits may require impractical amounts time. 
Fortunately, some of the most popular classification and regression tree methods, such as CART, Random Forest, and boosting with trees, use splits on one variable.
Recursive partitions based on splitting the data with one variable at a time yield much simpler cases. 
Splitting the data based on one variable yields a partition of the domain into rectangular boxes.
However, we provide some details for checking intersections or computing volumes regions for trees based on multivariate splits since this may be useful for some applications.

\subsubsection{Geometry of multivariate splits: hyperplanes and polytopes}\label{sec:Preliminaries}

\begin{figure}
	\includegraphics[width=9cm]{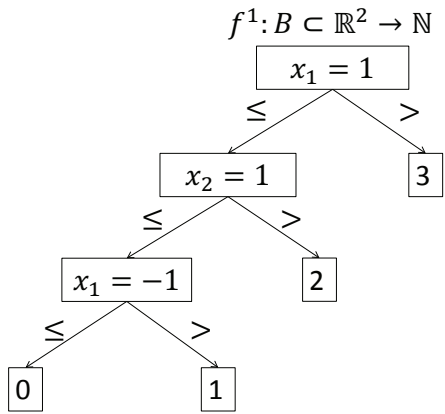}
	\includegraphics[width=9cm]{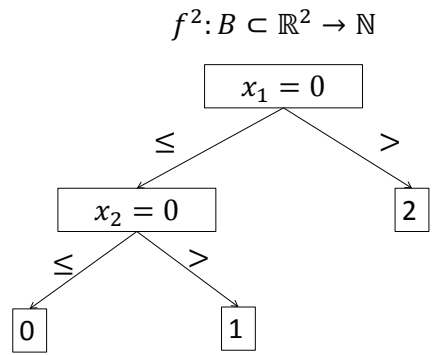}
	\includegraphics[width=9cm]{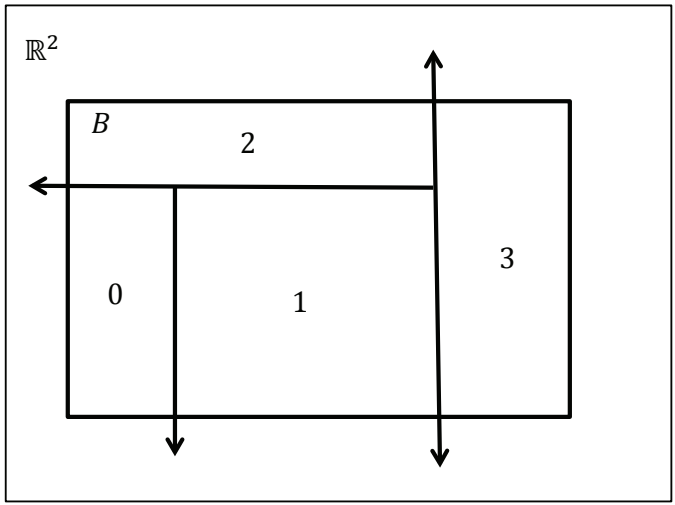}
	\includegraphics[width=9cm]{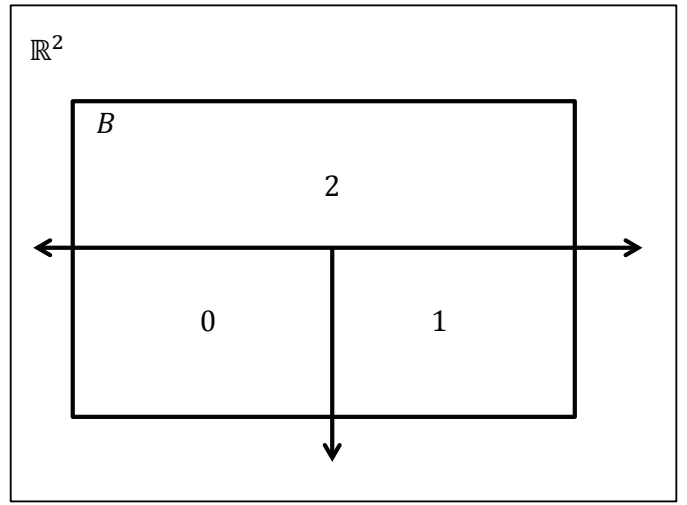}
	\caption{Two functions mapping polyhedral regions of two dimensional real Euclidean space, $\mathbb{R}^2$, defined by recursive partitioning, to the natural numbers, $\mathbb{N}$. }
	\label{fig:RecursivePartitionFunctionsExample}
\end{figure}

Given a scalar $b\in \mathbb{R}$, and a vector, $c$, in $n$-dimensional real Euclidean space, $ \mathbb{R}^n$, a linear equality, $c' x = b$, defines a hyperplane $H=\{x\in \mathbb{R}^n| c' x =b\}$. A hyperplane can be used to define two complementary regions, $H^>=\{x| c' x > b\}$ and $H^\leq=\{x|c' x \leq b\}$, called its upper open half-space, and its lower half-space, respectively. 

A polyhedron is a subset $S \subset \mathbb{R}^n$ which is defined by intersections of half-spaces and open half-spaces. 
A polyhedron can be divided into two complementary polyhedra contained in the complementary half-spaces of an intersecting hyperplane.


A recursive partition is a plane tree, $T$, with a root, $r$, a node set $V$ which naturally partitions into a set of interior nodes, $I$, a set of terminal nodes, $L$, called leafs.
Each leaf associated with a natural number, $n_l \in \mathbb{N}$
Each interior node, $v \in I$, associated with a hyperplane, $H_v \subset \mathbb{R}^n$.
The root and each interior node are associated with a left daughter and a right daughter, $v_{\leq} \in V$, and $v_{>} \in V$, respectively.
Data at $v$ are split, as follows.
Data in the upper half-space $H^>_{v}$ are partitioned to the right daughter and data in the lower half-space $H_v^\leq$ are partitioned to the left daughter.
The hyperplane associated with $v$ is oriented to intersect with the polyhedral region defined by the half-spaces along the path from the root, $r$, to $v$. 
Thus a recursive partition is a tuple of a plane tree and a set of hyperplanes associated with its nodes, $(T,\{H_v|v \in T\})$ which are assumed to intersect in this fashion.

Each point $x \in \mathbb{R}^n$ is contained in one of the polyhedral regions defined by the paths from the root to each leaf.
That is for all $x \in \mathbb{R}^n$, there is a unique leaf node, $l \in L$, such that, $x \in S(r,l)$. 
Thus, a recursive partition defines a polyhedral subdivision of Euclidean space.

Mapping each point $x \in \mathbb{R}^n$ to $n_l$ for $l$ such that $x \in S(r,l)$ defines a function $f:\mathbb{R}^n \to \mathbb{N}$. 
We assume that recursive partitions are used to define functions on a box $B=\{x\in \mathbb{R}^n|l_i \leq x_i \leq u_i\}$. 
For example consider the functions $T^1:B\subset \mathbb{R}^2\to \mathbb{N}$ and $T^2:B\subset \mathbb{R}^2 \to \mathbb{N}$ in Fig. \ref{fig:RecursivePartitionFunctionsExample} which are defined by recursive partitions.

\subsubsection{Tests for intersection of a Hyperplane and a Polyhedron}\label{sec:IntersectionTests}
Let $S$ be a polyhedral region in $\mathbb{R}^n$ defined by the intersection of half-spaces, that is $H^+_1 \cap \ldots \cap H^+_k$ where $H_i=\{x \in \mathbb{R}^n|c'_i x = b_i\} $.
Let $H=\{x \in \mathbb{R}^n | c'x=b\}$ be a hyperplane.
Does $H$ intersect $S$? If not, then is $S$ in $H^+$ or $H^-$?

Answering the questions of whether or not a hyperplane intersects a polytope is related to the problem in linear programming, of determining a set of minimal constraints for bounding a polytope. 
Due to the prevalence of linear programming this question has been investigated previously. 
A simple method is presented here, and finding the most efficient method available in the literature will be a topic of further research.

A linear program can be used to determine if $H$ intersects $S$. 
\begin{align}
\textrm{max } & c'x \label{LP1}\\
\textrm{s.t. } & c'_i x \leq b_i \forall i=1,\ldots,k \label{LP2}\\
& c'x \leq b+1 \label{LP3}
\end{align}

If the linear program defined by (\ref{LP1}-\ref{LP3}) is infeasible then the polytope $S$ is inside $H^+$. In this case the test should be conducted with the signs elements of $c$ reversed so that the polytope $S$ is contained inside the half-space defined by $c'x\leq b+1$. 
Otherwise the polyhedron is contained inside of $H^-$. 

Let $x^*$ be an optimal solution to the linear program defined by (\ref{LP1}-\ref{LP3}). If $c'x^*<b$ then the hyperplane $H$ does not intersect the polyhedron $S$. On the other hand if $c'x^*\geq b$ then the hyperplane $H$ and the polyhedron $S$ intersect.

\begin{figure}
	\includegraphics[width=9cm]{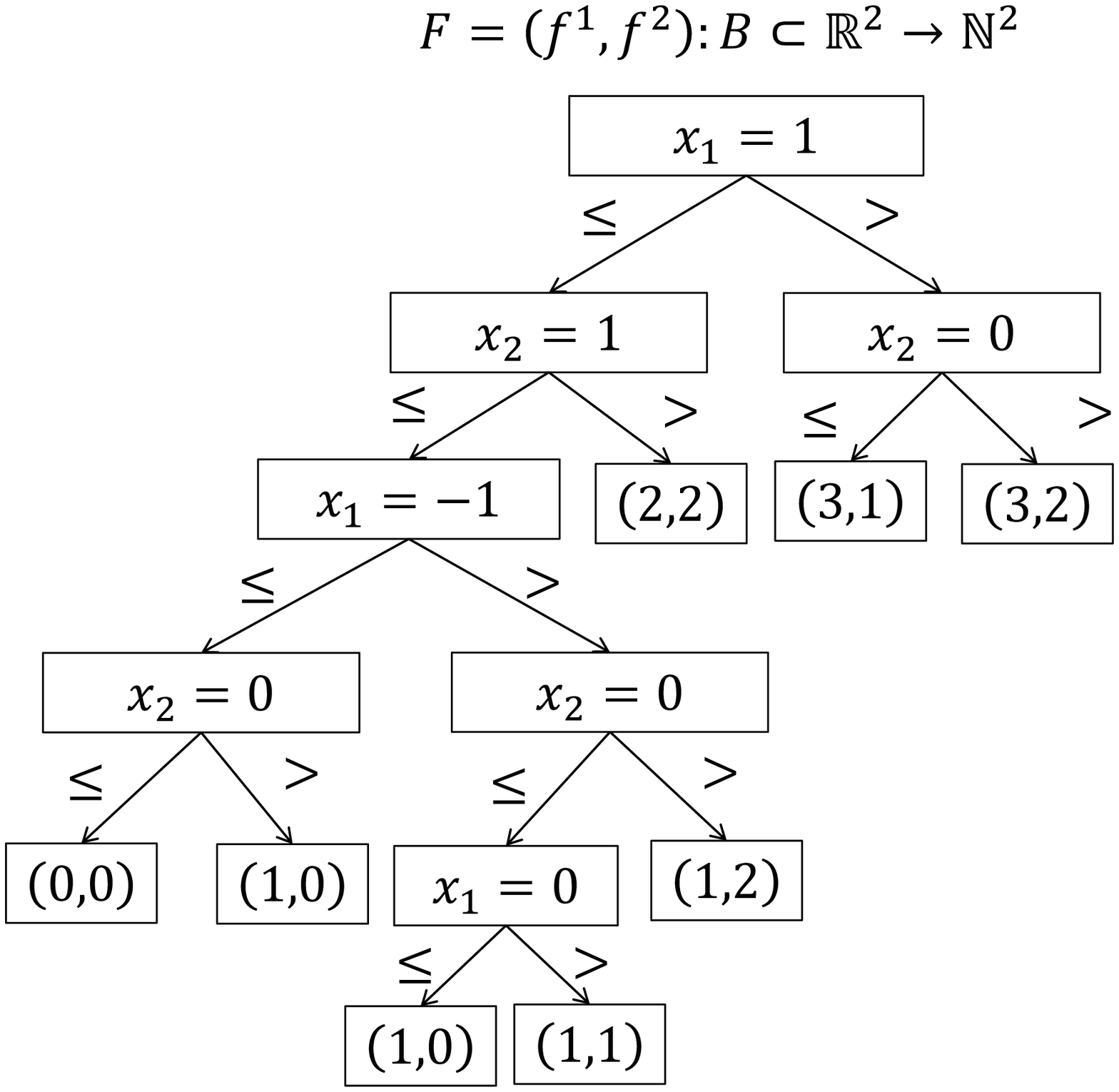}
	\includegraphics[width=9cm]{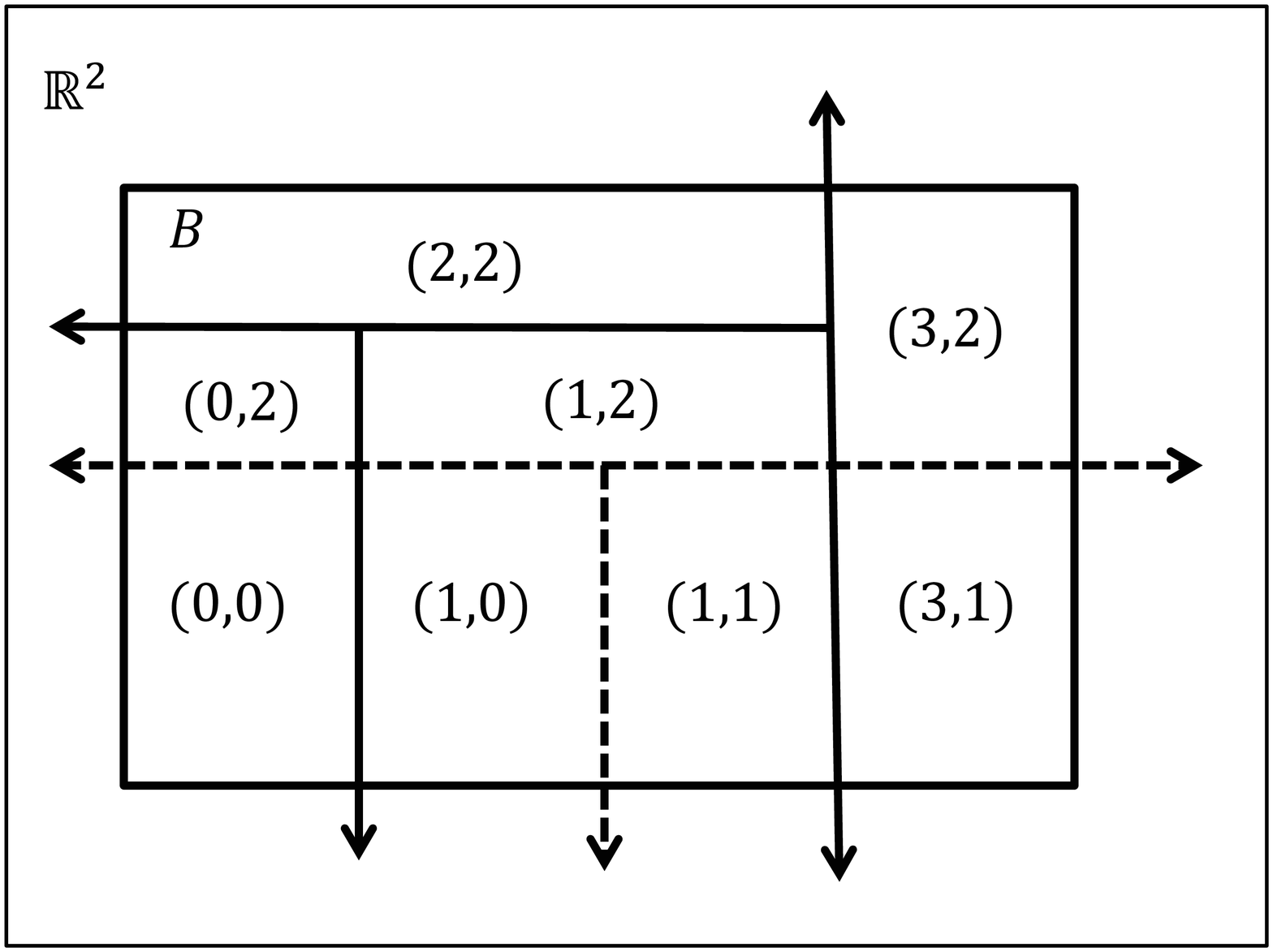}
	\caption{Representation of two recursive partitions, $T^1$ and $T^2$, shown in Fig. \ref{fig:RecursivePartitionFunctionsExample}, as a single recursive partition $F=(T^1,T^2)$ mapping from the same domain, $B$, to a range which is the product of the range of $T^1$ and the range of $T^2$. }
	\label{fig:CombinedRecursivePartition}
\end{figure}

\begin{figure}
	\includegraphics[width=9cm]{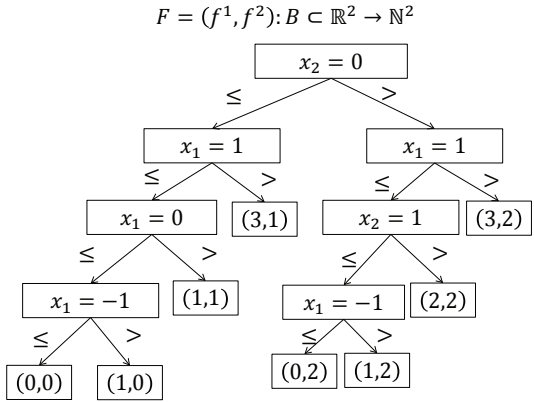}
	\caption{An alternate representation of $F=(T^1,T^2)$ (Figure \ref{fig:CombinedRecursivePartition}).}
	\label{fig:CombinedRecursivePartitionAlternate}
\end{figure}

\section{Concluding remarks and further research directions}
We presented a novel algorithm for computing a single tree which represents multiple recursive partition functions.
This algorithm facilitates quantifying the degree of difference or similarity between pairs of recursive partition functions.

Ensembles of trees are generally regarded as block-boxes for making predictions.
However, is it feasible to simplify an ensemble of trees to just a few trees which have a similar level of predictive power?
Although the algorithm presented in this paper could be used to combine many trees into a single tree, it may yield a tree which has many nodes, and would therefore be to large to comprehend entirely.
Methods for identifying clusters in ensembles of trees based on correlations or distances between trees
could be useful for building a smaller ensemble of a core of essential trees.
We leave these questions for further research.

\bibliographystyle{apalike}
\bibliography{C:/Users/ss3245/Documents/BibTex/library,C:/Users/ss3245/Documents/BibTex/general-stats,C:/Users/ss3245/Documents/BibTex/CARTEnsemblesNoMendeley}

\end{document}